\documentclass[sigconf]{acmart}
\AtBeginDocument{%
  }

\usepackage{hyperref}  
\usepackage{amsthm}
\usepackage{amsmath}
\usepackage{xcolor}
\usepackage[table]{xcolor} 
\usepackage{booktabs}
\usepackage{multirow}
\newtheorem{theorem}{Theorem}
\newtheorem{remark}{Remark}

\definecolor{LightBlue}{RGB}{247, 252, 255}
\definecolor{LightRed}{RGB}{254, 248, 248}
\definecolor{ggreen}{HTML}{cc0200}
\definecolor{gred}{HTML}{4C9F26}
\definecolor{mypurple}{RGB}{97,90,183}
\definecolor{red}{HTML}{c21313} 
\definecolor{mygreen}{RGB}{198,220,195}
\definecolor{myred}{RGB}{135,3,3}
\newcommand{\uaa}[1]{\scriptsize\textcolor{ggreen}{\footnotesize $\uparrow$}{\color{ggreen}#1}}
\newcommand{\daa}[1]{\scriptsize\textcolor{gred}{\footnotesize $\downarrow$}{\color{gred}#1}}
\definecolor{deepgreen}{RGB}{0,100,0}

\begin{document}

\title{GRIT: Graph-Regularized Logit Refinement for Zero-shot Cell Type Annotation}



\author{Tianxiang Hu}
\affiliation{%
  \institution{Zhejiang University}
  \city{Hangzhou}
  \country{China}}
\email{tianxiang.23@intl.zju.edu.cn}

\author{Chenyi Zhou}
\affiliation{%
  \institution{Zhejiang University}
  \city{Hangzhou}
  \country{China}}
\email{chenyi.22@intl.zju.edu.cn}

\author{Jiaxiang Liu}
\affiliation{%
  \institution{Guangdong Institute of Intelligence Science and Technology}
  \city{Hengqin}
  \country{China}}
\email{liujiaxiang@gdiist.cn}

\author{Jiongxin Wang}
\affiliation{%
 \institution{Zhejiang University}
 \city{Hangzhou}
 \country{China}}
\email{jiongxin.23@intl.zju.edu.cn}

\author{Ruizhe Chen}
\affiliation{%
  \institution{Zhejiang University}
  \city{Hangzhou}
  \country{China}}
\email{ruizhe.21@intl.zju.edu.cn}

\author{Haoxiang Xia}
\affiliation{%
  \institution{Zhejiang University}
  \city{Hangzhou}
  \country{China}}
\email{haoxiang.22@intl.zju.edu.cn}

\author{Gaoang Wang}
\affiliation{%
  \institution{Zhejiang University}
  \city{Hangzhou}
  \country{China}}
\email{gaoangwang@intl.zju.edu.cn}

\author{Jian Wu}
\affiliation{%
  \institution{Zhejiang University}
  \city{Hangzhou}
  \country{China}}
\email{jianwu@intl.zju.edu.cn}

\author{Zuozhu Liu}
\affiliation{%
  \institution{Zhejiang University}
  \city{Hangzhou}
  \country{China}}
\email{zuozhuliu@intl.zju.edu.cn}

\renewcommand{\shortauthors}{Hu et al.}


\begin{abstract}
Cell type annotation is a fundamental step in the analysis of single-cell RNA sequencing (scRNA-seq) data. In practice, human experts often rely on the structure revealed by principal component analysis (PCA) followed by $k$-nearest neighbor ($k$-NN) graph construction to guide annotation. While effective, this process is labor-intensive and does not scale to large datasets. Recent advances in CLIP-style models offer a promising path toward automating cell type annotation. By aligning scRNA-seq profiles with natural language descriptions, models like LangCell enable zero-shot annotation. While LangCell demonstrates decent zero-shot performance, its predictions remain suboptimal. In this paper, we propose a principled inference-time paradigm for zero-shot cell type annotation (GRIT) which bridges the scalability of pre-trained foundation models with the structural robustness relied upon in human expert annotation workflows. Specifically, we enforce local consistency of the zero-shot CLIP logits over the task-specific PCA-based $k$-NN graph. We evaluate our approach on 14 annotated human scRNA-seq datasets from 4 distinct studies, spanning 11 organs and over 200,000 single cells. Our method consistently improves zero-shot annotation accuracy, achieving accuracy gains of up to 10\%. Further analysis showcase the mechanism by which GRIT effectively propagates correct signals through the graph, pulling back mislabeled cells toward more accurate predictions. The method is training-free, model-agnostic, and serves as a simple yet effective plug-in for enhancing zero-shot cell type annotation.
\end{abstract}

\keywords{Cell type annotation, scRNA-seq, CLIP, Zero-shot classification}


\maketitle

\section{Introduction}
Single-cell RNA sequencing (scRNA-seq) technologies have enabled high-resolution profiling of cellular heterogeneity across diverse tissues and biological conditions \cite{pliner2019supervised,jovic2022single}. A critical step in the analysis pipeline is cell type annotation—assigning biologically meaningful labels to individual cells—which forms the basis for downstream interpretation. Traditionally, this process is performed through a semi-automatic workflow \cite{Pollen2015, Chu2016, quan2023annotation, clarke2021tutorial, pasquini2021automated, hu2023cellmarker}, combining dimensionality reduction (e.g., principal  component analysis), clustering algorithms (e.g., Leiden or Louvain), and manual inspection of marker gene expression. However, this approach is time-consuming and increasingly impractical as scRNA-seq datasets grow in size and complexity.

Recent advances in deep learning have enabled the development of powerful representation learning frameworks for scRNA-seq data. Models such as scBERT \cite{yang2022scbert}, scGPT \cite{cui2024scgpt}, Geneformer \cite{cui2024geneformer}, and scFoundation \cite{hao2024large} are pretrained on millions of single-cell profiles and support transfer learning for a wide range of downstream tasks. These foundation models offer a unified and scalable approach to encoding single-cell gene expression into informative embeddings using general-purpose neural architectures. Building on this foundation, LangCell \cite{zhao2024langcell} adopts Geneformer—a high-performing single-cell foundation model—as its cell encoder within a CLIP-style framework that aligns single-cell embeddings with natural language descriptions of cell identities. By training on paired scRNA-seq profiles and cell-type text annotations, LangCell uniquely enables zero-shot cell type annotation. Given the success of CLIP-style and multimodal foundation models in other domains when trained at scale, there is strong reason to believe that such architectures—when further developed and trained on higher quality scRNA-seq and text corpora—can evolve into powerful tools for automating and augmenting expert-level biological annotation.

\begin{figure}[t!]
\centering
\includegraphics[width=0.46\textwidth]{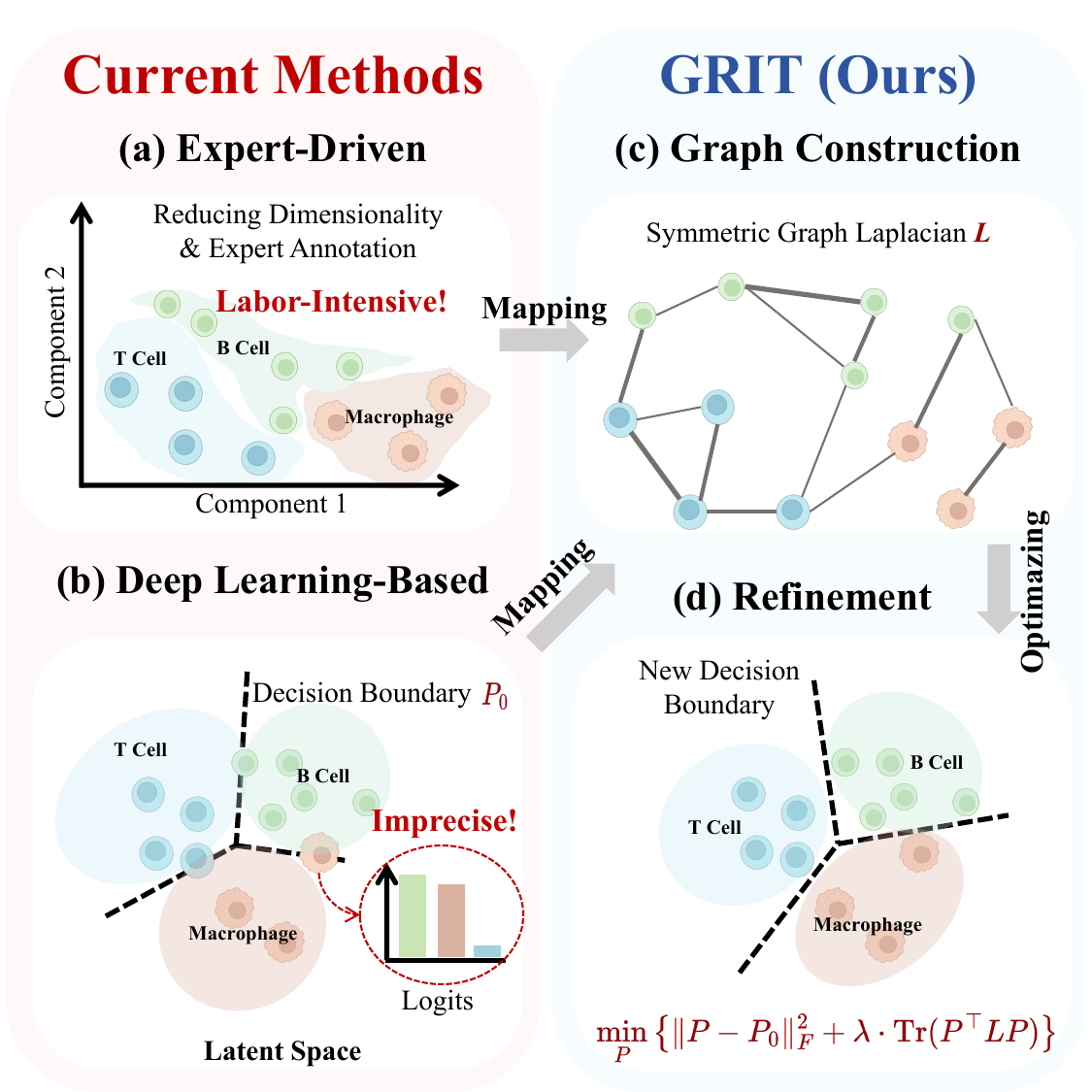}
\caption{GRIT Overview. (\textbf{a},\textbf{b}) Existing approaches rely on expert-driven labeling or deep learning models, which can be labor-intensive or imprecise. (\textbf{c}) Construct a PCA-based $k$-NN graph, with each node initialized by the logits predicted by a deep learning model. (\textbf{d}) GRIT refines these initial predictions by solving a graph-regularized optimization problem that promotes local consistency across the $k$-NN graph.} 
\label{fig:method}
\vspace{-1em}
\end{figure}

\begin{figure}[t!]
\centering
\includegraphics[width=0.46\textwidth]{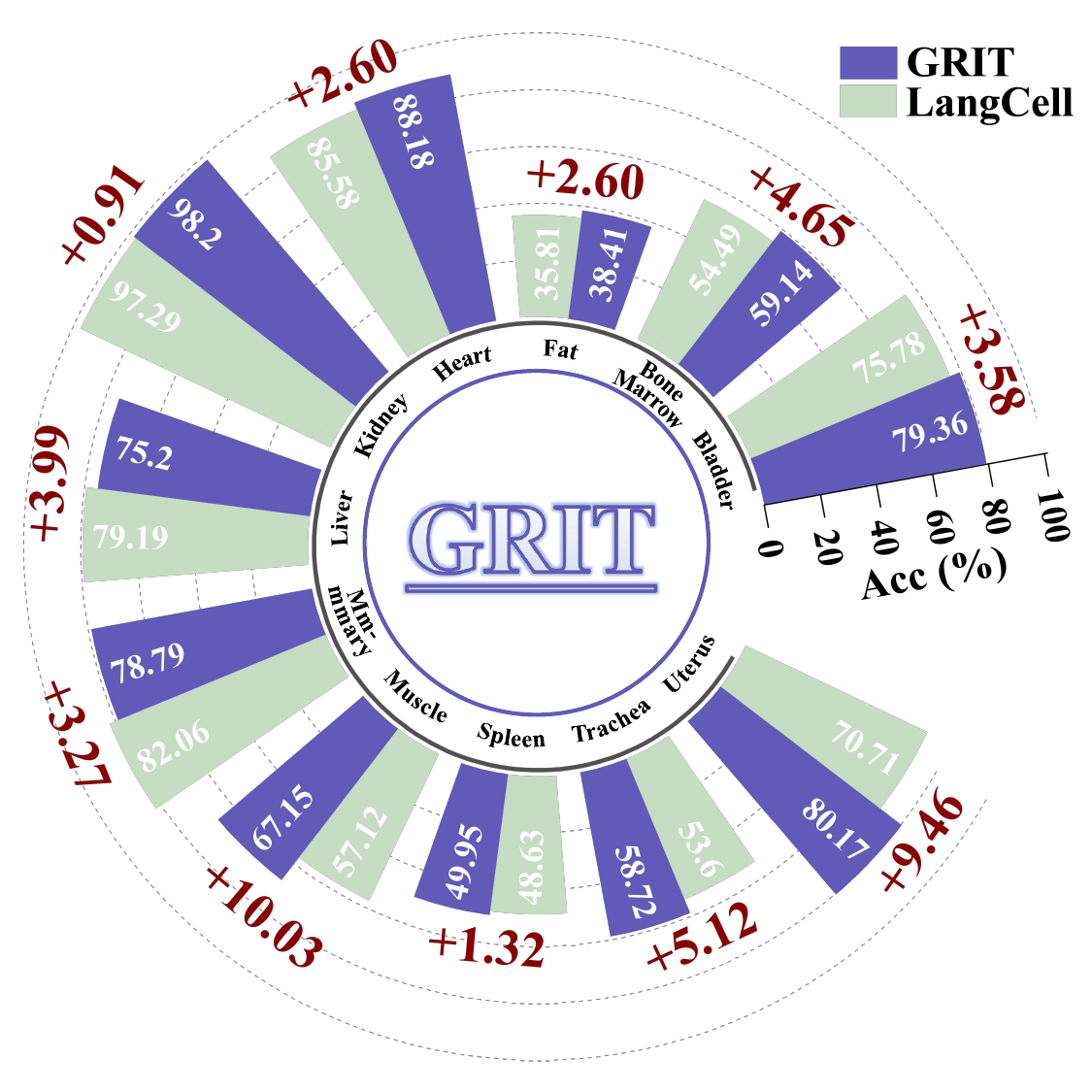}
\caption{Performance of GRIT on zero-shot cell type annotation across 11 organ-specific scRNA-seq datasets. Each segment shows the accuracy achieved by GRIT (\textcolor{mypurple}{blue}) and the baseline LangCell (\textcolor{mygreen}{green}), along with the accuracy gain labeled in \textcolor{myred}{red}. GRIT consistently improves performance over LangCell across all datasets, with accuracy gains up to 10\%. Detailed results and analysis are provided in Section~\ref{sec:main results}.} 
\label{fig:method-performance}
\vspace{-1em}
\end{figure}

An important observation in the context of single-cell annotation is that human experts frequently rely on the structure revealed by the PCA-based $k$-NN graph. This suggests that, for a given scRNA-seq dataset, we already possess a reliable task-relevant structure among cells. While models like LangCell enable zero-shot annotation by aligning single-cell representations with natural language description representations, they are not explicitly trained to preserve this structural information. As a result, the predictions may be locally inconsistent with the $k$-NN graph which serves as an informative structural reference in expert-guided annotation workflows. This limitation motivates a hybrid strategy: refining pre-trained models' zero-shot predictions by enforcing local consistency over the task-relevant $k$-NN graph. In doing so, we aim to combine the scalability provided by the pre-trained model with the structural robustness of the PCA-based $k$-NN graph. It is natural to consider existing methods such as label propagation or graph-based regularization. However, many of these approaches \cite{wang2019knowledge, zhu2003semi, zhou2003learning} require access to at least a subset of ground truth labels which are unavailable in the zero-shot setting. Others \cite{chen2024clustering, velivckovic2018deep, kipf2016semi} focus on enhancing graph neural network training, which may be unnecessary in our case. These limitations call for a label-free, inference-time refinement approach tailored to the zero-shot annotation setting.

In this paper, we propose a graph-regularized logit refinement method (GRIT) for enhancing zero-shot cell type annotation in scRNA-seq data. The core idea is to leverage the scalability of pretrained foundation models while correcting their predictions using domain-specific geometric consistency, without requiring any additional training (\autoref{fig:method}). Concretely, we first apply a CLIP-style model—LangCell in our case—to perform zero-shot annotation on scRNA-seq data, resulting in initial prediction logits. We then construct a PCA-based $k$-NN graph to reflect the underlying structure commonly trusted by human experts. Finally, we solve a graph-regularized optimization problem to refine the logits, encouraging smoothness over the graph while remaining close to the original predictions. This refinement is applied entirely at inference time and serves as a lightweight, principled postprocessing step. We rigorously analyze our approach and validate its effectiveness through extensive experiments on 14 annotated human scRNA-seq datasets, collected from 4 independent studies and spanning 11 organs and over 200,000 single cells. Experimental results show that our refinement consistently improves zero-shot performance, with nearly all cases showing gains in accuracy and macro F1 score. Our main contributions are:

\begin{itemize}
    \item We introduce a principled framework that integrates pretrained foundation models with human expert annotation practices through inference-time graph regularization of prediction logits. Our method is simple, training-free, and effective, serving as a plug-and-play solution for zero-shot cell type annotation.
    
    \item We provide a theoretical proof that graph-regularized optimization can effectively refine predictions, under the condition the initial logits are reasonably accurate.
    
    \item Extensive zero-shot cell type annotation experiments using LangCell across 14 human scRNA-seq datasets demonstrate consistent performance improvements through the application of our graph-regularized logit refinement.
\end{itemize}

\section{Related Work}


\subsection{Representation Learning for scRNA-seq}
The increasing availability of single-cell RNA sequencing (scRNA-seq) data has spurred the development of machine learning approaches for learning meaningful representations from high dimensional gene expression profiles \cite{olsen2018introduction,ziegenhain2017comparative,jovic2022single,papalexi2018single,zhang2021single,van2023applications,slovin2021single,hwang2018single,haque2017practical}. Traditional analysis pipelines often begin with dimensionality reduction techniques such as PCA or t-SNE, followed by clustering and manual marker gene inspection \cite{conard2023spectrum}. While these methods remain effective for visualizing structure and guiding expert annotation, recent work has turned to deep learning—particularly transformer-based architectures—to model complex nonlinear relationships across genes and enhance scalability to large, diverse datasets. Models such as \textsc{scBERT}\cite{yang2022scbert}, \textsc{scGPT}\cite{cui2024scgpt}, \textsc{Geneformer}\cite{cui2024geneformer}, and \textsc{scFoundation}\cite{hao2024large} leverage large-scale pretraining on millions of scRNA-seq profiles to support transfer learning across downstream tasks including cell type annotation.

\subsection{Language-Cell Alignment and Zero-Shot Cell Type Annotation}
Advances in multimodal representation learning aim to align heterogeneous data modalities within a shared semantic space, which is most prominently in the vision--language domain. A representative line of work demonstrates that contrastive or alignment-based pretraining can effectively bridge modality-specific representations with natural language supervision at scale. Concretely, general-purpose multimodal models such as \textsc{CLIP}~\cite{radford2021learning}, \textsc{BLIP}~\cite{li2022blip}, \textsc{LLaVA}~\cite{liu2023visual} and \textsc{Flamingo}~\cite{alayrac2022flamingo}, together with their domain-specific counterparts~\cite{wang2022medclip,stevens2024bioclip,li2023llava,moor2023med,jiang2025hulu}, demonstrate that aligning modality-specific encoders with large language models not only supports effective zero-shot transfer, but also enables instruction following and open-ended multimodal reasoning across diverse tasks.


Inspired by the success of these frameworks, recent work has begun to explore the alignment between cellular measurements and natural language. A representative example is \textsc{LangCell}~\cite{zhao2024langcell}, which introduces a CLIP-style contrastive pretraining framework to align scRNA-seq profiles with curated, ontology-derived textual descriptions. In this framework, \textsc{Geneformer}~\cite{cui2024geneformer} is adopted as the cell encoder, while a BERT-based model is used for text encoding. While LangCell demonstrates promising generalization in zero-shot settings like zero-shot cell type annotations, several limitations remain. First, neither its cell encoder nor the multimodal model itself explicitly leverages graph-based structural information during pretraining. Second, its inference mechanism—like other CLIP-style retrieval paradigms \cite{radford2021learning,liu2024vpl,kpl2025liu}—can be sensitive to suboptimal prompt design and modality biases. We introduce a post-hoc logit refinement approach that exploits graph-structured consistency among cells to enhance prediction accuracy at inference time.



\section{Methodology}

In this section, we present our graph-regularized refinement method for enhancing zero-shot cell type annotation in scRNA-seq data. The core idea is to improve the initial prediction logits produced by a pretrained model via encouraging local consistency among cells, guided by structural information captured in a PCA-based $k$-NN graph. We begin by providing an analysis demonstrating that, when the initial logits are reasonably accurate, applying graph regularization is guaranteed to improve prediction performance. We then introduce the full pipeline of our proposed method, GRIT, for zero-shot cell type annotation.

\subsection{Theoretical Analysis}

\begin{theorem}[Graph Regularized Logit Refinement Improves Predictions]
Given a symmetric graph Laplacian $L \in \mathbb{R}^{n \times n}$ constructed from an adjacency matrix $A \in \mathbb{R}^{n \times n}$, let $P_0 \in \mathbb{R}^{n \times c}$ denote the initial class logits over the $n$ nodes, and let $P^* \in \mathbb{R}^{n \times c}$ denote the ground-truth logits. Consider the following graph-regularized optimization problem:
\[
\hat{P}_\lambda := \arg\min_{P} \left\{ \|P - P_0\|_F^2 + \lambda \cdot \mathrm{Tr}(P^\top L P) \right\}.
\]

Suppose the following condition holds:
\[
\langle P_0 - P^*,\, L P_0 \rangle > 0.
\]

Then there exists a sufficiently small regularization parameter $\lambda > 0$ such that:
\[
\|\hat{P}_\lambda - P^*\|_F^2 < \|P_0 - P^*\|_F^2.
\]
\end{theorem}

\begin{proof}
The objective function is convex and quadratic, and the unique minimizer admits the closed-form solution:
\[
\hat{P}_\lambda = (I + \lambda L)^{-1} P_0.
\]
Define the function \( f(\lambda) := \|\hat{P}_\lambda - P^*\|_F^2 \). Since \( \hat{P}_\lambda \) is smooth in \( \lambda \), \( f(\lambda) \) is also continuously differentiable. 
At \( \lambda = 0 \), we have \( \hat{P}_0 = P_0 \), so:
\[
f(0) = \|P_0 - P^*\|_F^2.
\]
To study how \( f(\lambda) \) changes near 0, compute the derivative using the chain rule. Let \( A_\lambda := (I + \lambda L)^{-1} \), so \( \hat{P}_\lambda = A_\lambda P_0 \). Then:

\begin{align*}
f'(\lambda) 
  &= \frac{d}{d\lambda} \|A_\lambda P_0 - P^*\|_F^2 \nonumber \\
  &= -2 \cdot \mathrm{Tr}\left[
       (A_\lambda P_0 - P^*)^\top A_\lambda L A_\lambda P_0
     \right].
\end{align*}
At \( \lambda = 0 \), we have \( A_0 = I \), so:
\begin{align*}
f'(0) 
  &= -2 \cdot \mathrm{Tr} \left[ (P_0 - P^*)^\top L P_0 \right] \\
  &= -2 \langle P_0 - P^*, L P_0 \rangle.
\end{align*}
By assumption, this quantity is strictly negative. Since \( f(\lambda) \) is differentiable and \( f'(0) < 0 \), there exists some \( \lambda > 0 \) sufficiently small such that:
\[
f(\lambda) < f(0),
\]
i.e.,
\[
\|\hat{P}_\lambda - P^*\|_F^2 < \|P_0 - P^*\|_F^2.
\]
\end{proof}

\begin{remark}
The condition $\langle P_0 - P^*, L P_0 \rangle > 0$ captures a meaningful alignment between the prediction error and the graph structure. It implies that the residuals in $P_0$ are not arbitrary but exhibit disagreement that is structured according to the graph Laplacian. In other words, when a prediction is incorrect, its neighbors tend to disagree in directions that the regularizer penalizes. This condition is satisfied when $P_0$ is informative yet imperfect—e.g., capturing coarse structure but missing fine details. In such cases, the graph enables error correction by propagating reliable information across neighbors. This formalizes the intuition that logit refinement is effective when initial predictions are reasonable and graph reflects meaningful similarity.
\end{remark}

\subsection{The GRIT Method}

\begin{figure*}[tb]
\centering
\includegraphics[width=0.90\textwidth]{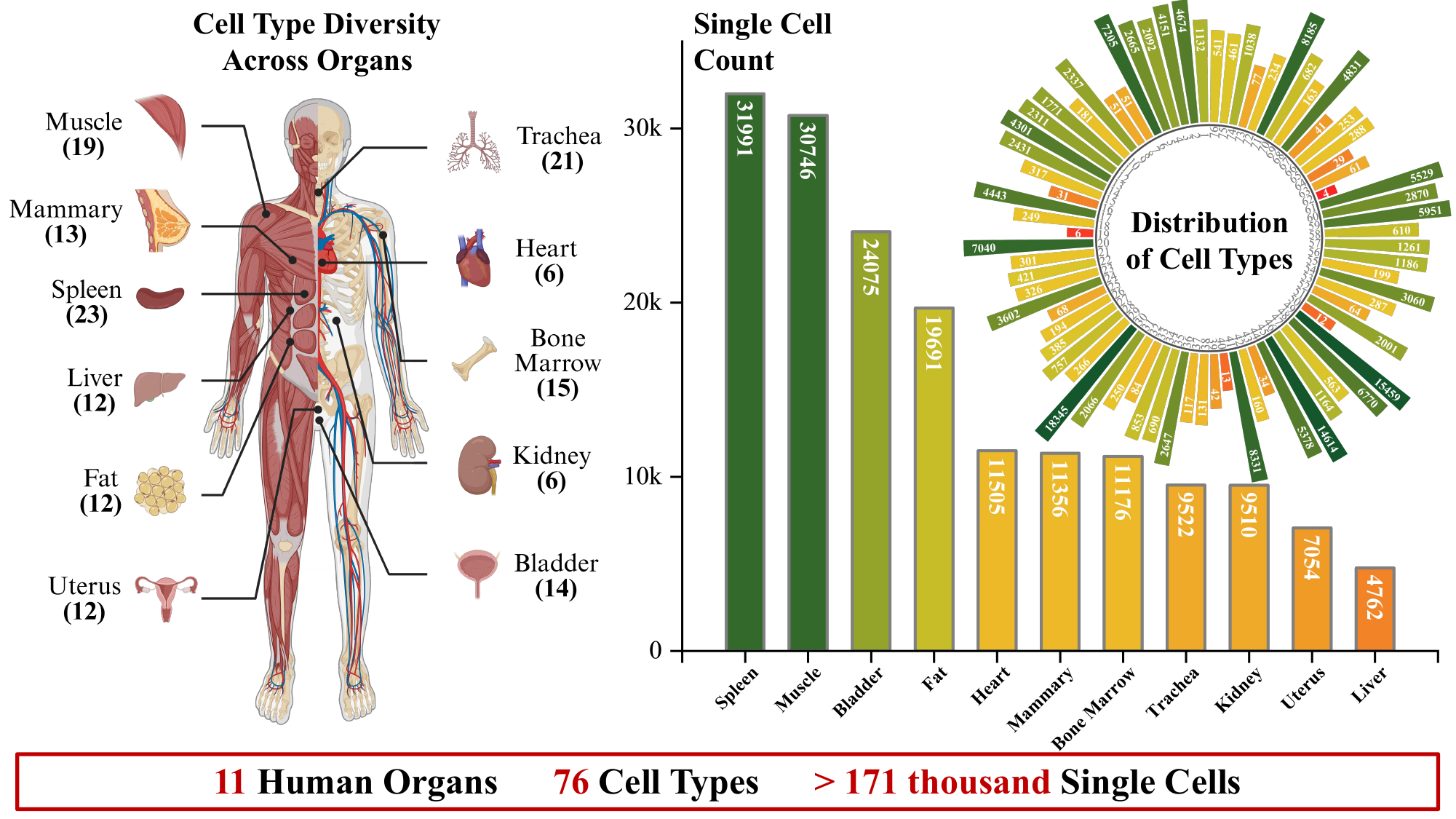}
\caption{Overview of the human scRNA-seq datasets used in our main experiment. They span 11 human organs, 76 annotated cell types, and over 171{,}000 single cells. The anatomical illustration summarizes cell type diversity across organs. The bar chart reports single-cell counts per organ. The circular plot visualizes the distribution of all cell types, indexed from 1 to 76 for clarity. Full cell type names corresponding to these indices are listed in Appendix.}
\label{fig:dataset}
\vspace{-1em}
\end{figure*}

\begin{table*}[t]
\centering
\caption{Performance comparison on zero-shot cell type annotation across 11 scRNA-seq datasets from the \textit{Tabula Sapiens} project. We evaluate LangCell, Majority Vote, Logit Average, and GRIT under the default hyperparameter setting, using Accuracy and Macro F1 (\%). GRIT consistently outperforms the baseline LangCell across all datasets. The average row (\textbf{Avg.}) summarizes overall performance, with gains shown as colored subscripts. }
\label{tab:main results}
\begin{tabular}{c|cccccccc}
\toprule
\multirow{2}{*}{\textbf{Dataset}} 
& \multicolumn{2}{c}{\textbf{LangCell}} 
& \multicolumn{2}{c}{\textbf{Majority Vote}}
& \multicolumn{2}{c}{\textbf{Logit Average}}
& \multicolumn{2}{c}{\textbf{GRIT}} \\
\cmidrule(lr){2-3} \cmidrule(lr){4-5} \cmidrule(lr){6-7} \cmidrule(lr){8-9}
& Accuracy & Macro F1
& Accuracy & Macro F1
& Accuracy & Macro F1
& Accuracy & Macro F1 \\
\midrule
\textbf{Bladder}     
& 75.78 & 52.57 
& 79.79 $\textcolor{red}{\uparrow}$ 
& 52.67 $\textcolor{red}{\uparrow}$ 
& 78.80 $\textcolor{red}{\uparrow}$ 
& 51.62 $\textcolor{gred}{\downarrow}$ 
& 79.36 $\textcolor{red}{\uparrow}$ 
& 53.32 $\textcolor{red}{\uparrow}$ \\
\textbf{Bone Marrow}       
& 54.49 & 36.76 
& 61.61 $\textcolor{red}{\uparrow}$ 
& 37.34 $\textcolor{red}{\uparrow}$ 
& 59.86 $\textcolor{red}{\uparrow}$ 
& 35.34 $\textcolor{gred}{\downarrow}$ 
& 59.14 $\textcolor{red}{\uparrow}$ 
& 37.59 $\textcolor{red}{\uparrow}$ \\
\textbf{Fat}     
& 35.81 & 32.70 
& 38.95 $\textcolor{red}{\uparrow}$ 
& 30.07 $\textcolor{gred}{\downarrow}$ 
& 39.24 $\textcolor{red}{\uparrow}$ 
& 29.22 $\textcolor{gred}{\downarrow}$ 
& 38.41 $\textcolor{red}{\uparrow}$ 
& 34.56 $\textcolor{red}{\uparrow}$ \\
\textbf{Heart}     
& 85.58 & 70.97 
& 88.16 $\textcolor{red}{\uparrow}$ 
& 71.73 $\textcolor{red}{\uparrow}$ 
& 89.40 $\textcolor{red}{\uparrow}$ 
& 71.98 $\textcolor{red}{\uparrow}$ 
& 88.18 $\textcolor{red}{\uparrow}$ 
& 73.32 $\textcolor{red}{\uparrow}$ \\
\textbf{Kidney}   
& 97.29 & 84.60 
& 93.85 $\textcolor{gred}{\downarrow}$ 
& 52.59 $\textcolor{gred}{\downarrow}$ 
& 93.55 $\textcolor{gred}{\downarrow}$ 
& 50.73 $\textcolor{gred}{\downarrow}$ 
& 98.20 $\textcolor{red}{\uparrow}$ 
& 85.03 $\textcolor{red}{\uparrow}$ \\
\textbf{Liver}   
& 75.20 & 52.23 
& 77.89 $\textcolor{red}{\uparrow}$ 
& 48.84 $\textcolor{gred}{\downarrow}$ 
& 77.99 $\textcolor{red}{\uparrow}$ 
& 48.27 $\textcolor{gred}{\downarrow}$ 
& 79.19 $\textcolor{red}{\uparrow}$ 
& 56.85 $\textcolor{red}{\uparrow}$ \\
\textbf{Mammary} 
& 78.79 & 50.07 
& 81.83 $\textcolor{red}{\uparrow}$ 
& 42.15 $\textcolor{gred}{\downarrow}$ 
& 82.08 $\textcolor{red}{\uparrow}$ 
& 46.45 $\textcolor{gred}{\downarrow}$ 
& 82.06 $\textcolor{red}{\uparrow}$ 
& 52.66 $\textcolor{red}{\uparrow}$ \\
\textbf{Muscle}    
& 57.12 & 31.35 
& 70.61 $\textcolor{red}{\uparrow}$ 
& 32.09 $\textcolor{red}{\uparrow}$ 
& 72.60 $\textcolor{red}{\uparrow}$ 
& 31.79 $\textcolor{red}{\uparrow}$ 
& 67.15 $\textcolor{red}{\uparrow}$ 
& 35.50 $\textcolor{red}{\uparrow}$ \\
\textbf{Spleen}  
& 48.63 & 30.40 
& 48.26 $\textcolor{gred}{\downarrow}$ 
& 26.36 $\textcolor{gred}{\downarrow}$ 
& 48.56 $\textcolor{gred}{\downarrow}$ 
& 25.77 $\textcolor{gred}{\downarrow}$ 
& 49.95 $\textcolor{red}{\uparrow}$ 
& 30.44 $\textcolor{red}{\uparrow}$ \\
\textbf{Trachea}     
& 53.60 & 39.34 
& 61.65 $\textcolor{red}{\uparrow}$ 
& 37.38 $\textcolor{gred}{\downarrow}$ 
& 61.36 $\textcolor{red}{\uparrow}$ 
& 37.25 $\textcolor{gred}{\downarrow}$ 
& 58.72 $\textcolor{red}{\uparrow}$ 
& 40.34 $\textcolor{red}{\uparrow}$ \\
\textbf{Uterus}     
& 70.71 & 44.07 
& 82.56 $\textcolor{red}{\uparrow}$ 
& 45.05 $\textcolor{red}{\uparrow}$ 
& 83.98 $\textcolor{red}{\uparrow}$ 
& 46.01 $\textcolor{red}{\uparrow}$ 
& 80.17 $\textcolor{red}{\uparrow}$ 
& 47.57 $\textcolor{red}{\uparrow}$ \\
\midrule
\cellcolor{LightBlue}\textbf{Avg.}       
& \cellcolor{LightBlue}\textbf{66.64} 
& \cellcolor{LightBlue}\textbf{47.73} 
& \cellcolor{LightBlue}\textbf{71.38} \uaa{4.74} 
& \cellcolor{LightBlue}\textbf{43.30} \daa{4.43} 
& \cellcolor{LightBlue}\textbf{71.58} \uaa{4.94} 
& \cellcolor{LightBlue}\textbf{43.13} \daa{4.60} 
& \cellcolor{LightBlue}\textbf{70.96} \uaa{4.32} 
& \cellcolor{LightBlue}\textbf{49.74} \uaa{2.01} \\
\bottomrule
\end{tabular}
\end{table*}

\begin{table*}[t]
\centering
\caption{Logit consistency before and after the GRIT refinement across 11 organs. GRIT consistently reduces logit inconsistency with respect to the underlying graph structure.}
\label{tab:logit_consistency}
\resizebox{\textwidth}{!}{%
\begin{tabular}{cccccccccccc}
\toprule
\textbf{Logit Consistency} & \textbf{Bladder} & \textbf{Bone Marrow} & \textbf{Fat} & \textbf{Heart} & \textbf{Kidney} & \textbf{Liver} & \textbf{Mammary} & \textbf{Muscle} & \textbf{Spleen} & \textbf{Trachea} & \textbf{Uterus} \\
\midrule
\textbf{$P_0^T L P_0$ ($k=15$)} & 60.54 & 25.10 & 40.09 & 35.85 & 44.41 & 15.58 & 29.32 & 81.48 & 58.48 & 17.94 & 17.67 \\
    \textbf{$\hat{P_{\lambda}}^T L \hat{P_{\lambda}}$ ($k=15$)} & 16.42 & 6.91 & 10.57 & 9.35 & 12.17 & 4.28 & 8.25 & 21.64 & 15.09 & 4.94 & 4.95 \\
\textbf{$P_0^T L P_0$ ($k=20$)} & 60.98 & 25.37 & 40.28 & 36.01 & 44.62 & 15.67 & 29.70 & 82.21 & 58.71 & 18.07 & 17.79 \\
    \textbf{$\hat{P_{\lambda}}^T L \hat{P_{\lambda}}$ ($k=20$)} & 16.75 & 7.10 & 10.75 & 9.51 & 12.44 & 4.35 & 8.51 & 22.12 & 15.34 & 5.04 & 5.05 \\
\bottomrule
\end{tabular}
}
\end{table*}

\begin{table*}[t]
\centering
\caption{Wall-clock time (minutes) required for GRIT graph regularization with $k$=15 and $k$=20 across the 11 datasets used in the main experiments, together with the number of cells (n) and the number of cell types (c) for each dataset.}
\label{tab:computational_cost}
\resizebox{\textwidth}{!}{%
\begin{tabular}{cccccccccccc}
\toprule
\textbf{-} & \textbf{Bladder} & \textbf{Bone Marrow} & \textbf{Fat} & \textbf{Heart} & \textbf{Kidney} & \textbf{Liver} & \textbf{Mammary} & \textbf{Muscle} & \textbf{Spleen} & \textbf{Trachea} & \textbf{Uterus} \\
\midrule
$n$ & 24, 075 & 11, 176 & 19, 691 & 11, 505 & 9, 510 & 4, 762 & 11, 356 & 30, 746 & 31, 991 & 9, 522 & 7, 054 \\
    $c$ & 14 & 15 & 12 & 6 & 6 & 12 & 13 & 19 & 23 & 21 & 12 \\
Clock-time ($k$=15) & 0.76 & 0.11 & 0.33 & 0.18 & 1.06 & 0.03 & 0.12 & 10.27 & 5.79 & 0.20 & 0.11 \\
    Clock-time ($k$=20) & 1.03 & 0.16 & 0.40 & 0.21 & 1.73 & 0.04 & 0.15 & 16.25 & 7.20 & 0.24 & 0.13 \\
\bottomrule
\end{tabular}
}
\end{table*}

We propose a three-stage framework GRIT for zero-shot cell type annotation in scRNA-seq data. Given an unlabeled dataset of $n$ cells with gene expression profiles $\{x_i\}_{i=1}^n$ and a set of $c$ cell type descriptions $\{t_j\}_{j=1}^c$, our method consists of: (1) obtaining initial logits via a CLIP-style model, (2) constructing a $k$-NN graph constructed in PCA space and (3) refining the initial logits using graph-regularized optimization (\autoref{fig:method}).

\paragraph{Step 1: Initial Prediction via CLIP-Style Model.}
We first adopt a CLIP-style model, LangCell, to obtain the initial prediction logits $P_0 \in \mathbb{R}^{n \times c}$, where each row corresponds to a cell and each column to a candidate cell type. For each cell $x_i, P_0(i) =  \mathrm{model}\left(x_i, \{t_j\}_{j=1}^c\right)$, where $\mathrm{model}(\cdot, \cdot)$ produces a probability distribution reflecting the alignment between the input cell and each cell type candidate. LangCell uses these logits for zero-shot prediction via: $g(x_i) = \arg\max_j \left\{ P_0(i) \right\}, \quad \forall i \in \{1, \dots, n\},$ enabling automatic cell type annotation without any task-specific training or fine-tuning. The logits $P_0$ serve as the input to our refinement procedure.

\paragraph{Step 2: Graph Construction.}
To model the relational structure among cells, we construct a $k$-NN graph based on a low-dimensional representation of the input data. Specifically, each cell $x_i$ is first preprocessed following standard scRNA-seq analysis procedures. The resulting gene expression matrix is standardized and projected onto a $d$-dimensional space using principal component analysis (PCA), yielding a reduced representation $x_i^{\mathrm{PCA}}$ for each cell. Using these PCA-reduced features, we build a $k$-NN graph $\mathcal{G} = (\mathcal{V}, \mathcal{E})$ by connecting each cell to its $k$ nearest neighbors in Euclidean space. The graph is symmetrized so that an undirected edge exists between cell $i$ and cell $j$ if either is among the $k$ nearest neighbors of the other. Let $A \in \mathbb{R}^{n \times n}$ denote the adjacency matrix of this graph, the unnormalized graph Laplacian is then defined as $L = D - A$, where $D$ is the diagonal degree matrix with $D_{ii} = \sum_j A_{ij}$. This Laplacian $L$ captures the local geometry of the dataset and is used in the subsequent logit refinement step.

\paragraph{Step 3: Graph-regularized Logit Refinement.}
After obtaining the initial logits $P_0$ and the $k$-NN graph represented by the adjacency matrix $A$ and its corresponding Laplacian $L$, we refine the logits by solving a graph-regularized optimization problem that promotes local consistency: $\hat{P_{\lambda}} = \arg\min_{P} \{ \|P - P_0\|_F^2 + \lambda\, \mathrm{Tr}(P^\top L P) \},$ where $\lambda > 0$ controls the strength of the regularization. This objective encourages the refined logits $P$ to remain close to the initial predictions $P_0$ while being smooth with respect to the graph structure encoded by $L$. The solution to this convex quadratic problem has a closed form and can be computed efficiently by solving a linear system: $\hat{P_{\lambda}} = (I + \lambda L)^{-1} P_0.$ GRIT then uses the refined logits to perform zero-shot cell type annotation: $h(x_i) = \arg\max_j \left\{ \hat{P}(i) \right\}$.

\section{Experiments}

\subsection{Datasets and Setup}
We compile 11 scRNA-seq datasets from the \textit{Tabula Sapiens} project~\cite{the2022tabula}, spanning organs including \textbf{Bladder}, \textbf{Bone Marrow}, \textbf{Fat}, \textbf{Heart}, \textbf{Kidney}, \textbf{Liver}, \textbf{Mammary}, \textbf{Muscle}, \textbf{Spleen}, \textbf{Trachea}, and \textbf{Uterus} (see \autoref{fig:dataset}). In total, the datasets encompass 76 cell types and 171,383 single cells. In addition, we collect two public peripheral blood mononuclear cell datasets (\textbf{PBMC10k} \cite{gayoso2022python} and \textbf{PBMC368k} \cite{zheng2017massively}), and a peripheral cortex dataset (\textbf{Peripheral Cortex} \cite{siletti2023transcriptomic}). These datasets cover 18 cell types and 33, 700 single cells. All datasets include cell type annotations curated by domain experts. Dataset statistics is detailed in Appendix.

We focus on the zero-shot cell type annotation task on these datasets and employ LangCell, a pretrained CLIP model adapted for scRNA-seq and text alignment, to generate the initial soft label logits $P_0$. Following standard preprocessing procedures in Scanpy \cite{wolf2018scanpy} (see Appendix), a PCA-based $k$ nearest neighbor ($k$-NN) graph is constructed for each dataset to capture cell similarity relationships. Unless otherwise stated, we use $\alpha=0.2$ for LangCell, 50 principal components for PCA, and $k=15$ for $k$-NN as our default hyperparameters. The initial logits $P_0$ are refined using GRIT with a default hyperparameter of $\lambda=1$. 

For comparison, we evaluate GRIT against LangCell and two intuitive baselines that directly combine the initial logits with the PCA-based $k$-NN graph. \textbf{(1) Majority Vote:} the predicted cell type is determined by the majority vote among the predictions of a cell and its neighbors $\arg\max_j
\sum_{i' \in \mathcal{N}_k(i)}
\mathbb{I}\!\left( g(x_{i'}) = j \right)$. \textbf{(2) Logit Average:} the logits of a cell and its neighbors are averaged to produce a refined logit, and the final prediction is obtained using the refined logit via $\arg\max_j
\frac{1}{|\mathcal{N}_k(i)|}
\sum_{i' \in \mathcal{N}_k(i)}
P_0(i')$. Here we denote $\mathcal{N}_k(i)$ the set of $k$ nearest neighbors of cell $i$ in the PCA-based $k$-NN graph (including cell $i$ itself). Performance of all methods is assessed using accuracy and macro F1 score (see Appendix), evaluating both overall and class-balanced predictive performance. 

\subsection{Main Results}
\label{sec:main results}

Constructing a $k$-NN graph on principal components is a standard step in scRNA-seq analysis pipelines \cite{guo2003knn, zhang2017learning}. It is common to reduce gene expression profiles to 50 principal components and set $k$ to 15 or 20 when building the $k$-NN graph. For example, the widely used toolkits Scanpy \cite{wolf2018scanpy} and Seurat \cite{pereira2021asc} adopt 50 principal components by default, with $k=15$ and $k=20$, respectively. To align with established practice among domain experts, we evaluate GRIT under these commonly used configurations.

As shown in \autoref{tab:main results}, GRIT consistently improves performance over the LangCell baseline across all 11 organ datasets, with accuracy improvements up to 10.03\% (Muscle) and macro F1 improvements up to 4.62\% (Liver) (see Appendix table for results with $k=20$). Moreover, the results under $k=15$ and $k=20$ are similar, indicating strong robustness to the choice of $k$. \autoref{tab:logit_consistency} presents the logit consistency before and after GRIT refinement for both $k$ settings, quantified by $P_0^T L P_0$ and $\hat{P}_{\lambda}^T L \hat{P}_{\lambda}$, respectively. GRIT consistently reduces logit inconsistency with respect to the graph structure, resulting in refined logits that better align with the $k$-NN graph used for cell type classification. 

These findings support our central hypothesis that exploiting the $k$-NN graph structure—routinely employed by practitioners for cell type annotation—enables meaningful refinement of the initial prediction logits produced by decent models, resulting in significant and stable improvements in predictive performance.

The refinement step solves a sparse linear system with a fixed coefficient matrix for each class using scipy.sparse.linalg.spsolve. Detailed computational cost are reported in \autoref{tab:computational_cost}. In particular, on the most computationally demanding dataset (Muscle, $n$=30,746, $c$=19), the default GRIT refinement completes in 10.27 ($k$=15) and 16.25 ($k=20$) minutes of wall-clock time.


We compare GRIT with two intuitive methods that directly combine the initial logits with the $k$-NN graph: Majority Vote and Logit Average. Although these approaches can increase prediction accuracy, they reduce the macro F1 score in more than half of all evaluated cases. Specifically, Majority Vote reduces the macro F1 on 6 out of 11 datasets for $k=15$ and on 8 out of 11 datasets for $k=20$. Logit Average reduces the macro F1 on 8 out of 11 datasets for $k=15$ and on 9 out of 11 datasets for $k=20$. This pattern indicates that such approaches tend to favor aggregate accuracy at the expense of class-wise performance balance, resulting in less balanced predictions. Notably, in the Kidney dataset, where the initial LangCell predictions are already strong (97.29\% accuracy and 84.60\% macro F1), both methods lead to substantial performance degradation. Majority Vote reduces accuracy to 93.85\% (-3.44\%) and macro F1 to 52.59\% (-32.01\%) for $k=15$, and to 93.88\% (-3.41\%) accuracy and 52.56\% (-32.04\%) macro F1 for $k=20$. Logit Averaging yields 93.55\% (-3.74\%) accuracy and 50.73\% (-33.87\%) macro F1 for $k=15$, and 93.62\% (-3.67\%) accuracy and 51.29\% (-33.31\%) macro F1 for $k=20$. This observation suggests that these heuristics are unstable and can be detrimental when the initial predictions are already of high quality.

\subsection{Empirical Analysis of  $\lambda$}
In the Methodology section, we analyzed the behavior of the function $f(\lambda)$ which quantifies the quality of the refined logits $\hat{P}_{\lambda}$. Under a mild condition on the initial logits $P_0$, we showed that $f(\lambda)$ decreases over a small right-hand interval near $\lambda = 0$. Note that by definition, smaller values of $f(\lambda)$ correspond to higher-quality refined logits, whereas larger values indicate poorer-quality refinement. Motivated by this insight, we conduct a systematic sweep over $\lambda$ values near zero. Specifically, we evaluate $\lambda \in \{10^{-5}, 10^{-4}, 10^{-3}, 10^{-2}, \allowbreak\ 0.1,  0.2, 0.5, 1, 2, 5, 10, 20, 30, 40, 50, 60, 70, 80, 90, 100 \}$. \autoref{fig:lambda} (also see Appendix) illustrates how the prediction performance varies with $\lambda$ across all 11 human organ datasets, both individually and on average. Both the individual and average trend align with our analysis: performance improves as $\lambda$ increases from zero. Empirically, we find that moderate values in the range $\lambda \in (0, 5)$ enhance the overall performance. We adopt $\lambda = 1$ as the default. 

\begin{figure*}[tb]
\centering
\includegraphics[width=0.99\textwidth]{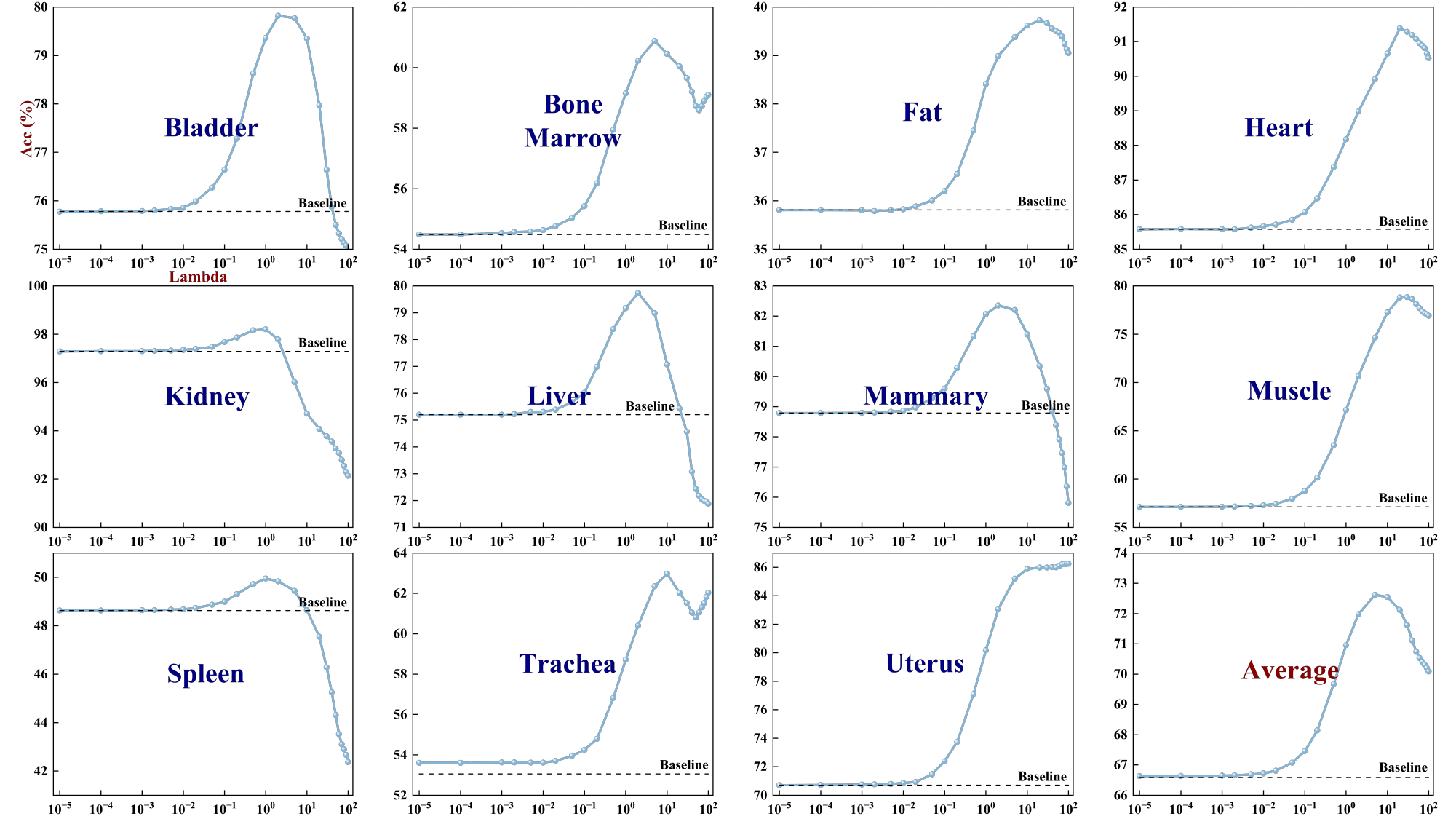}
\caption{Investigation of GRIT performance in the right-hand neighborhood of $\lambda = 0$ across all 11 organs from the \textit{Tabula Sapiens} project and their average. The x-axis denotes $\lambda$ values, and the y-axis reports logit performance measured by accuracy. Dashed lines indicate baseline accuracies achieved by LangCell.} 
\label{fig:lambda}
\end{figure*}

\begin{table}[t]
\centering
\caption{Performance comparison between LangCell and GRIT under varying $\alpha$ values and on additional datasets. We report Accuracy and Macro F1 (\%). Arrows indicate relative performance of GRIT compared to LangCell.}
\label{tab:grit_ablation_mainstyle}
\setlength{\tabcolsep}{4pt}  
\begin{tabular}{c|cccc}
\toprule
\multirow{2}{*}{\textbf{Setting}} 
& \multicolumn{2}{c}{\textbf{LangCell}} 
& \multicolumn{2}{c}{\textbf{GRIT}} \\
\cmidrule(lr){2-3} \cmidrule(lr){4-5}
& Accuracy & Macro F1
& Accuracy & Macro F1 \\
\midrule
$\alpha = 0.1$
& 63.60 & 46.01
& 66.87 $\textcolor{red}{\uparrow}$ 
& 47.66 $\textcolor{red}{\uparrow}$ \\

$\alpha = 0.2$
& 66.59 & 47.73
& 70.96 $\textcolor{red}{\uparrow}$ 
& 49.74 $\textcolor{red}{\uparrow}$ \\

$\alpha = 0.3$
& 67.39 & 48.32
& 71.48 $\textcolor{red}{\uparrow}$ 
& 49.78 $\textcolor{red}{\uparrow}$ \\

$\alpha = 0.5$
& 67.05 & 48.27
& 70.65 $\textcolor{red}{\uparrow}$ 
& 49.47 $\textcolor{red}{\uparrow}$ \\

\midrule
PBMC10k
& 86.52 & 83.91
& 88.43 $\textcolor{red}{\uparrow}$ 
& 87.57 $\textcolor{red}{\uparrow}$ \\

PBMC368k
& 84.86 & 83.17
& 87.41 $\textcolor{red}{\uparrow}$ 
& 86.16 $\textcolor{red}{\uparrow}$ \\

Peripheral Cortex
& 98.01 & 72.87
& 98.41 $\textcolor{red}{\uparrow}$ 
& 73.27 $\textcolor{red}{\uparrow}$ \\

\bottomrule
\end{tabular}
\end{table}

\begin{figure*}[t]
\centering
\includegraphics[width=0.80 \textwidth]{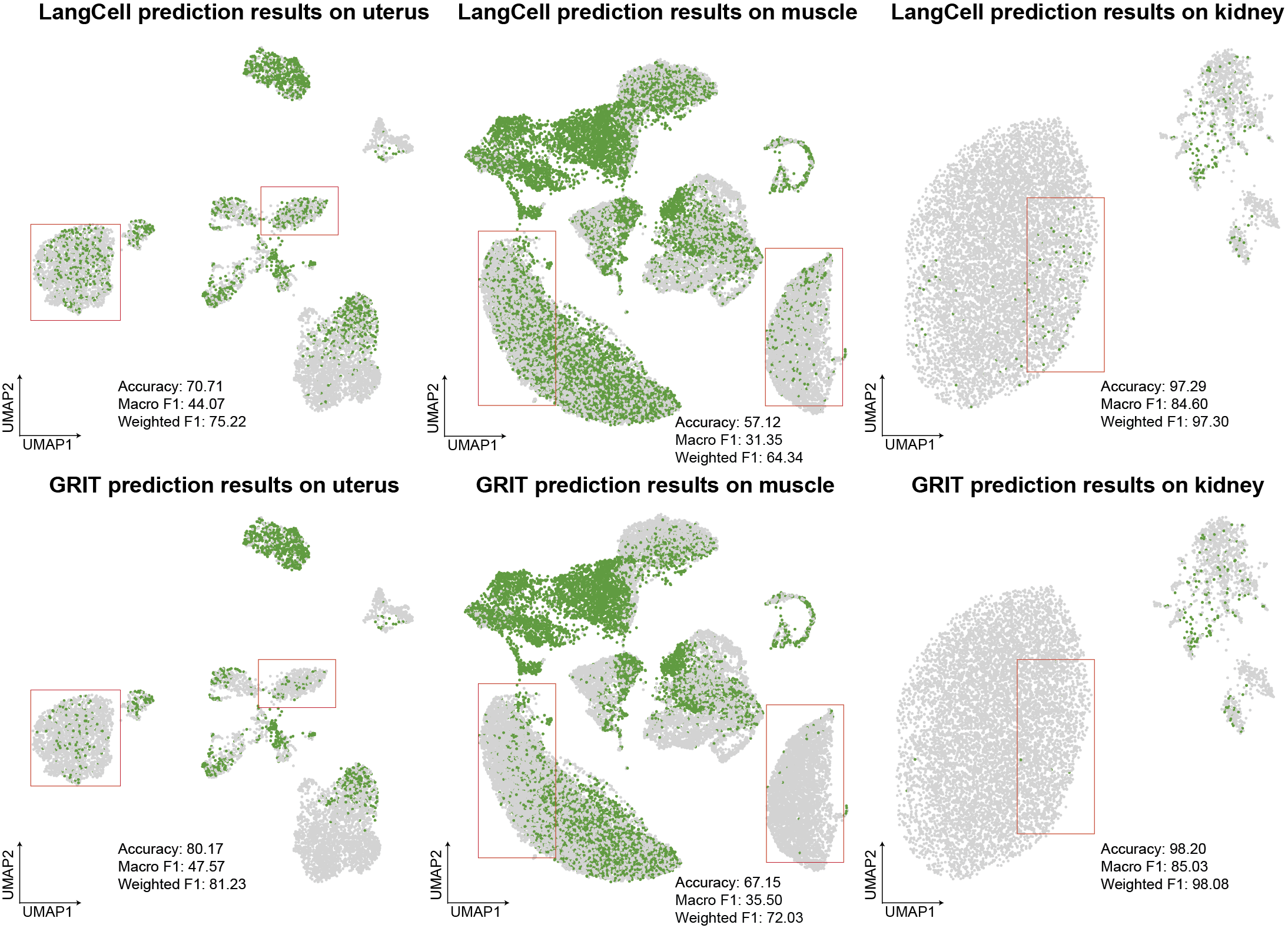}
\caption{UMAP visualization of scRNA-seq data of organ uterus (left), muscle (middle), and kidney (right). Each point represents a cell, colored by prediction correctness: gray indicates correct predictions, and green indicates incorrect ones. For each organ, the top panel shows LangCell zero-shot predictions, the bottom panel shows refined predictions from GRIT. Orange boxes indicate representative regions where GRIT provides clear improvements.
} 
\label{fig:case study}
\vspace{-1em}
\end{figure*}

\begin{figure*}[t]
\centering
\includegraphics[width=0.996\textwidth]{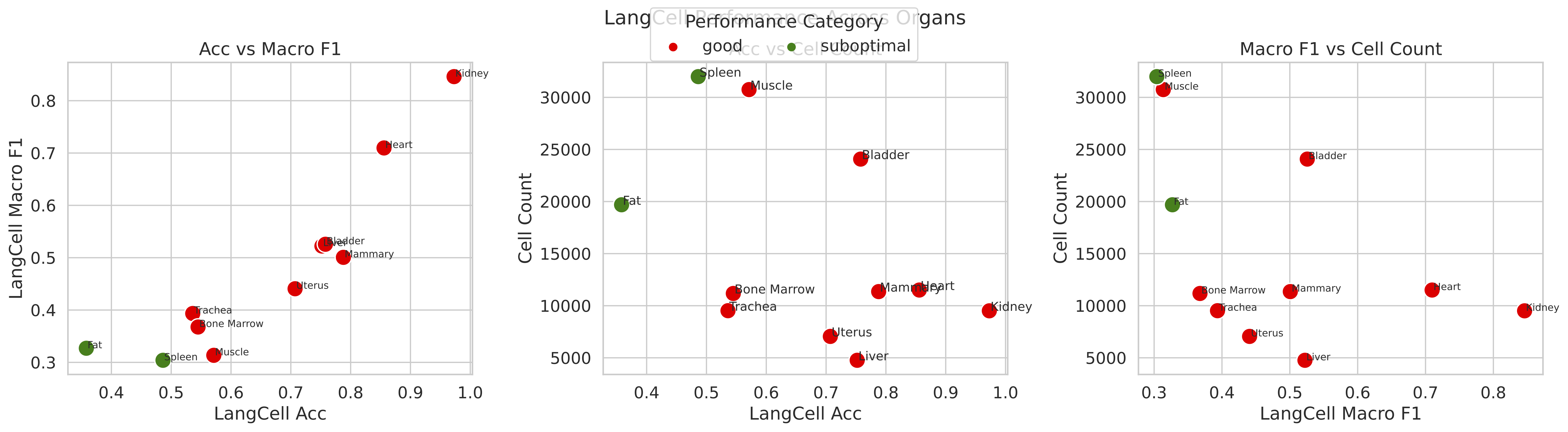}
\vspace{-1em}
\caption{GRIT performance are plotted against LangCell logit initial accuracy, macro F1, and total cell count for all 11 scRNA-seq datasets used in our main experiment.} 
\label{Figure_case_study}
\end{figure*}

\subsection{Robustness Across Different LangCell Hyperparameters and More Datasets}
Our main experiments, along with empirical analysis in the low-$\lambda$ regime, suggest that LangCell's zero-shot predictions are already sufficiently aligned with the graph structure for GRIT to provide effective improvements. To ensure a comprehensive evaluation, we further investigate a broader range of LangCell hyperparameter settings and datasets from studies in addition to the \textit{Tabula Sapiens} project to assess generalizability across diverse input conditions.

As reported in LangCell, the hyperparameter $\alpha \in \{0.1, 0.2, 0.3, 0.5\}$ yields competitive results, with $\alpha = 0.2$ recommended as the default (see Appendix tables for details). We adopt $\alpha = 0.2$ for baseline comparisons, but also evaluate GRIT under all four $\alpha$ settings to assess its robustness across a range of initializations. As shown in \autoref{tab:grit_ablation_mainstyle}, GRIT consistently improves zero-shot annotation performance, measured by accuracy and macro F1, across all $\alpha$ values. These results underscore two key observations: (1) current zero-shot models like LangCell are capable of producing logits that already reflect a latent task-relevant structure, which can be effectively enhanced through graph-based refinement. (2) GRIT is robust to variations in initial logit quality, demonstrating consistent gains across a spectrum of strong configurations. \autoref{tab:grit_ablation_mainstyle} also presents GRIT's performance on additional datasets including PBMC10k, PBMC368k, and Peripheral Cortex. GRIT consistently improves prediction performance across these datasets, demonstrating strong robustness to diverse data sources with varying levels of data noise.

\subsection{Visualization and Case Study}
We present three representative cases that exhibit clear performance improvements with GRIT, in order to illustrate the effect of the proposed refinement (\autoref{fig:case study}). In the UMAP visualizations, gray points denote correctly predicted cells, while green points indicate incorrectly predicted ones. As shown, the kidney dataset exhibits high initial prediction performance with LangCell, whereas the uterus and muscle datasets exhibit comparatively lower baseline accuracy. Nonetheless, across all cases, as long as a subset of correctly predicted cells exists within a local neighborhood, GRIT can propagate these signals to pull back mislabeled cells.

However, not all datasets benefit equally: in fat and spleen, GRIT slightly reduces macro F1 (see $k=20$ results in Appendix). To further explore this trend, we analyze how GRIT improvement varies across various initial prediction qualities and dataset sizes. The improvement is visualized using color, where red indicates a significant performance gain, and green denotes a marginal improvement. Specifically, we examine initial accuracy, macro F1, and total cell count. As shown in \autoref{Figure_case_study}, GRIT tends to be more effective when the initial accuracy and macro F1 are relatively high. Additionally, when the initial performance is poor, large cell count number can degrade results. This observation is intuitive: reliable initial predictions ensure that neighbors provide trustworthy information, and a greater number of neighbors increases the likelihood of correcting noisy predictions through refinement. For noisy initial logits, the same propagation may reinforce incorrect patterns.

It is important to note that GRIT is not designed to achieve 100\% accuracy. Rather, similar to other inference-time refinement methods \cite{qian2023intra, liu2025kpl}, it seeks to leverage complementary structural information in the data to bring the initial logits closer to their achievable performance limit. In our case, this is accomplished by promoting local consistency and semantic alignment with the $k$-NN graph. While the ultimate performance is bounded by the quality of the initial predictions, GRIT, like other inference-time refinement methods, introduces a principled and effective post hoc procedure that consistently yields measurable improvements.

\section{Limitations and Ethical Considerations}
All experiments in this work are conducted on publicly available and de-identified datasets that were released for research use. We do not collect new data, nor do we have access to any personally identifiable information. Informed consent was obtained by the original data collectors, and no additional consent procedures are required for this study. Accordingly, this study does not involve direct interaction with human participants.

\section{Conclusion}
We present GRIT, a biologically grounded, training-free refinement method for improving zero-shot cell type annotation in scRNA-seq data. GRIT is designed to be as a principled, plug-in post-processing step that can be applied to any pretrained model producing cell-level logits. While effective, GRIT assumes a reliable graph and reasonably accurate initial logits. Future works include developing adaptive $\lambda$ selection strategies and extending GRIT to multi-omics.




\section{GenAI Disclosure}
Generative AI tools were used in a limited capacity during the preparation of this manuscript, primarily for language polishing and minor editorial assistance. All methodological designs, algorithmic formulations, experimental analyses, and conclusions were conceived, implemented, and verified by the authors. The use of generative AI did not influence the reported results.

\bibliographystyle{ACM-Reference-Format}
\bibliography{references}

\end{document}